\documentclass[preprint]{elsarticle}
\usepackage[hmargin=3.65cm,vmargin=3.65cm]{geometry}

\usepackage[utf8]{inputenc}
\usepackage{amsthm,amsmath,amssymb}
\usepackage{graphicx,subfigure}
\usepackage{color}

\let\oldmarginpar\marginpar
\renewcommand\marginpar[1]{\-\oldmarginpar[\raggedleft\footnotesize #1]{\raggedright\footnotesize #1}}

\newcommand{\etal}{\emph{et al.}}
\newcommand{\ie}{\emph{i.e.}}
\newcommand{\eg}{\emph{e.g.}}
\newcommand{\define}[1]{\emph{#1}}
\newcommand{\resp}{resp.}

\newcommand{\N}{\ensuremath{\mathbb{N}}}
\newcommand{\ppO}{\ensuremath{\mathcal{O}}}
\newcommand{\NP}{\ensuremath{\mathsf{NP}}}
\newcommand{\FPT}{\ensuremath{\mathsf{FPT}}}
\newcommand{\XP}{\ensuremath{\mathsf{XP}}}
\newcommand{\W}[1]{\ensuremath{\mathsf{W}[\mathsf{#1}]}}
\newcommand{\rfpt}{\emph{fpt}}

\newcommand{\DS}{\textsc{Dominating set}}
\newcommand{\si}{\ensuremath{\sigma}}
\newcommand{\rh}{\ensuremath{\rho}}
\newcommand{\SR}{\ensuremath{[\si,\rh]}}
\newcommand{\SRDS}{\SR-\DS}
\newcommand{\pexists}{\ensuremath{\exists}}
\newcommand{\pmin}{\ensuremath{\min}}
\newcommand{\pmax}{\ensuremath{\max}}
\newcommand{\eSRDS}{\pexists\SR-\DS}
\newcommand{\eSRDSps}{\pexists\SR-\textsc{Dominating set with preselected vertices}}

\newcommand{\kSRDS}{$k$-\SRDS}
\newcommand{\kCDS}{$k$-\textsc{Capacitated dominating set}}

\newcommand{\CMSO}{\textrm{CMSO}}

\newcommand{\tw}{\ensuremath{\mathit{tw}}}
\newcommand{\poly}{\ensuremath{\mathrm{poly}}}

\newcommand{\gadgetname}[1]{\emph{#1}}
\newcommand{\gadget}[2]{\ensuremath{\mathcal{#1}_{#2}}}
\newcommand{\capacity}{\ensuremath{\mathrm{cap}}}
\newcommand{\dom}{\ensuremath{\mathrm{dom}}}
\newcommand{\Gap}{\ensuremath{\mathsf{\Gamma}}}
\newcommand{\gapz}{\ensuremath{\Gap_0}}
\newcommand{\gapm}{\ensuremath{\Gap_-}}
\newcommand{\gapp}{\ensuremath{\Gap_+}}
\newcommand{\ppdeg}{\ensuremath{\mathrm{deg}}}

\newcommand{\problemname}[1]{\textsc{#1}}
\newcommand{\problem}[3]{

  \medskip
  \noindent
  \problemname{#1}\\
  \textit{Input:} #2\\
  \textit{Question:} #3
  \medskip
}
\newcommand{\paramproblem}[4]{
  \medskip
  \noindent
  \problemname{#1}\\
  \textit{Input:} #2\\
  \textit{Parameter:} \ensuremath{#3}\\
  \textit{Question:} #4
  \medskip
}

\newtheorem{theorem}{Theorem}
\newtheorem*{theoremnn}{Theorem}
\newtheorem{lemma}{Lemma}
\newtheorem{corollary}{Corollary}

\title{W[1]-hardness of some domination-like problems parameterized by tree-width\tnoteref{t1}}

\author[ulb]{Mathieu Chapelle}
\tnotetext[t1]{Research supported by the French National Research Agency (ANR) project AGAPE. Most of this work has been done during the Ph.D. thesis of the author at LIFO, Université d'Orléans, France.}
\address[ulb]{Université Libre de Bruxelles, CP~212, B-1050 Brussels, Belgium}
\ead{mathieu.chapelle@univ-orleans.fr}

\begin{document}

\begin{abstract}
The concept of generalized domination unifies well-known variants of domination-like and independence problems, such as \textsc{Dominating set}, \textsc{Independent set}, \textsc{Perfect code}, \emph{etc}.
A generalized domination (also called \textsc{$[\sigma,\rho]$-Dominating Set}) problem consists in finding a subset of vertices in a graph such that every vertex is satisfied with respect to two given sets of constraints $\sigma$ and $\rho$.
Very few problems are known not to be $\mathsf{FPT}$ when parameterized by tree-width, as usually this restriction allows one to write efficient algorithms to solve the considered problems.
The main result of this article is a proof that for some (infinitely many) sets $\sigma$ and $\rho$, the problem \textsc{$\exists[\sigma,\rho]$-Dominating Set} is $\mathsf{W}[1]$-hard when parameterized by the tree-width of the input graph.
This contrasts with the current knowledge on the parameterized complexity of this problem when parameterized by tree-width, which had only been studied for finite and cofinite sets $\sigma$ and $\rho$ and for which it has been shown to be $\mathsf{FPT}$.
\end{abstract}



\maketitle



\section{Introduction}

\subsection{Motivation} 

Parameterized complexity is a recent theory introduced in the late 90's by Downey and Fellows (see \eg~\cite{DoFe99,FlGr06} for surveys).
This theory underlines the connection between a parameter (different from the usual size of the input) and the complexity of a given problem, and allows one to better study its complexity.
Many different problem-specific parameters can be considered, such as the maximum size of a desired solution, or the tree-width of the input graph.
A problem is said to be \FPT\ (fixed-parameter tractable) parameterized by a parameter $k$ if it can be solved in $\ppO(f(k) \cdot p(n))$ time, for some computable function $f$ and a polynomial $p$, where $n$ is the size of the input.
Parameterized intractable problems are at least \W{1}-hard, where \W{1} is one of the most important classes of parameterized complexity and believed to be strictly including the class \FPT\ (see \eg~\cite{DoFe99,FlGr06}).

In this article, we study the parameterized complexity of \emph{generalized domination} problems, also known as \eSRDS\ and introduced by Telle~\cite{Te94,Te94PhD}, when parameterized by the tree-width of the input graph.
Let $\si,\rh$ be two fixed subsets of $\N$ (throughout this paper, $\N$ denotes the set of nonnegative integers while $\N^*$ denotes the set of positive integers).
The problem is defined as follows:

\problem{\eSRDS}{A graph $G = (V,E)$.}{Is there a subset $D \subseteq V$ such that for every $v \in D$, $|N(v) \cap D| \in \si$, and for every $v \notin D$, $|N(v) \cap D| \in \rh$? If so, $D$ is called a \SR-dominating set.}

Notice that the sets $\si$ and $\rh$ are part of the definition of the problem, hence we can suppose that they are given as oracles deciding the membership of an integer to these sets.

It is well known that usual optimization problems such as \problemname{Minimum dominating set} (minimum \SR-dominating set with $\si = \N$ and $\rh = \N^*$) or \problemname{Maximum independent set} (maximum \SR-dominating set with $\si = \{0\}$ and $\rh = \N$) are \NP-hard.
When dealing with generalized domination, in many cases the problem of finding any \SR-dominating set is already \NP-hard.
Thus one usually considers first the problem of deciding the existence of such set in a given graph, and if relevant, the optimization problems \pmin-\SRDS\ and \pmax-\SRDS\ asking for an \SR-dominating set of minimum or maximum size respectively.
In this paper, we mainly consider existence problems unless otherwise stated.

Many well-known \NP-hard problems become efficiently tractable when restricted to graphs of bounded tree-width, and from the parameterized complexity point of view lots of them have turned out to be \FPT\ when parameterized by the tree-width of the input graph.
The decomposition into a tree-like structure of the input graph allows one to write algorithms which efficiently solve the considered hard problems.
A natural question is whether this tree-like structure can be used to solve every \SRDS\ problems in \FPT\ time parameterized by tree-width.

In this article, we show that for (infinitely) many cases of \si\ and \rh, the problem \eSRDS\ is \W{1}-hard even when restricted to graphs of bounded tree-width, giving a more accurate picture of the parameterized complexity of this problem when parameterized by tree-width.
We also prove (in Section~\ref{sec:general_complexity}) that this problem can be solved in $\ppO(n^{f(\tw)})$ time whenever \si\ and \rh\ are some \emph{reasonable} sets, hence justifying that this problem ought to be studied from the parameterized viewpoint.
Finally, note that by a result from Courcelle \etal~\cite{CoMaRo01}, it can be proven that \eSRDS\ is \FPT\ when parameterized by tree-width when \si\ and \rh\ are both ultimately periodic sets, and that an efficient algorithm can be obtained using automata and dynamic programming (see Section~\ref{sec:general_complexity}).

\subsection{Related work}

The \eSRDS\ problem, also known as locally checkable vertex subset and vertex partitioning problems, has been extensively studied since its introduction by Telle~\cite{Te94,Te94PhD} (see also \eg~\cite{BuTeVa13,CaPe12,GoVi08,RaSaSr08,RoBoRo09,TePr97,HeTe98}).

Several papers have considered the computational complexity of \eSRDS\ for some cases of $\si$ and $\rh$, on general graphs (see \eg~\cite{Te94,Te94PhD}), and on some classes of graphs such as bounded tree-width graphs~\cite{Te94,Te94PhD,TePr97}, bounded boolean-width graphs~\cite{BuTeVa13}, or chordal graphs~\cite{GoKr07}.

From the parameterized complexity point of view, it is well known that most of the usual existence domination-like problems are \W{1}-complete or \W{2}-complete on general graphs (see \eg~\cite{DoFe95-1,DoFe95-2}).
In an attempt to unify these results on \kSRDS\ parameterized by the maximum size $k$ of a \SR-dominating set, Cattanéo and Perdrix~\cite{CaPe12} have shown that it is in \W{2} for any recursive sets \si\ and \rh, while Golovach \etal~\cite{GoKrSu10} have shown that it is \W{1}-complete when \si\ and \rh\ are both finite sets.
Moreover, it is known that \eSRDS\ is \FPT, parameterized by tree-width, when \si\ and \rh\ are finite or cofinite sets (see \eg~\cite{TePr97,RoBoRo09}).

One may wonder whether every problem solvable in $\ppO\big(n^{\poly(\tw)}\big)$ parameterized by tree-width (\ie, in \XP, see \eg\ \cite{DoFe99,FlGr06}) is also solvable in \FPT\ time for the same parameter.
The answer is no, and some (few) parameterized problems are known to be \W{1}-hard when parameterized by tree-width (see \eg~\cite{DoLoSaVi08,FeFoLoRoSaTh07}).

\subsection{Our result}

We show that for (infinitely) many cases of $\si$ and $\rh$, the problem \eSRDS\ is \W{1}-hard when parameterized by the tree-width of the input graph.
For this purpose, we focus mainly on \si, and prove the following:

\begin{theorem}\label{thm:main}
Let $\si \subseteq \N$ be a set with arbitrarily large gaps between two consecutive elements, such that a gap of length at least $t$ is at distance $\poly(t)$ in \si (see Section~\ref{subsec:functions_on_sigma}), and let $\rh \subseteq \N$ be cofinite.
Suppose $\min \si \geq 1$ and $\min \rh \geq 2$.
Then the problem \eSRDS\ is \W{1}-hard when parameterized by the tree-width of the input graph.
\end{theorem}

Throughout this paper, the more-intuitive term \define{gap} will designate an excluded interval of consecutives integers in a subset of (non-negative) integers.

Many natural well-known infinite subsets of integers verify the condition on \si\ given above, \eg\ the positive powers of $\alpha \geq 2$ (\eg\ for $\alpha = 3$, the set $\{3,9,27,81,\ldots\}$), or the Fibonacci numbers ($\{1,2,3,5,8,13,\ldots\}$).
On the other hand, this result doesn't apply for infinite sets with bounded gaps, \eg\ the ultimately periodic sets (for which the problem is in fact \FPT, see Section~\ref{sec:general_complexity}), or the set of all positive integers excepted the multiple of $\alpha \geq 2$ (for \eg\ $\alpha = 3$, the set $\{1,2,4,5,7,8,10,\ldots\}$).

\bigskip

This article is organized as follows.
In Section~\ref{sec:preliminaries}, we recall some notions and definitions, and we give in Section~\ref{sec:general_complexity} some general results on the parameterized complexity of \eSRDS.
The proof of our main Theorem~\ref{thm:main} is then splitted in two steps.
Firstly in Section~\ref{sec:w1-hardness_first-step}, we reduce from \kCDS, one of the few problems known to be \W{1}-hard~\cite{DoLoSaVi08} when parameterized by the tree-width of the input graph (and the maximum size of a solution), to a variant of \eSRDS\ in which some vertices are inconditionally included in the \SR-dominating set.
Secondly in Section~\ref{sec:w1-hardness_second-step}, we reduce this variant to the \eSRDS\ problem, hence proving this latter to be \W{1}-hard when parameterized by tree-width.
We conclude this article with some open questions in Section~\ref{sec:conclusion}.

\section{Preliminaries}\label{sec:preliminaries}

We briefly recall in this section some notions and definitions used throughout this article.

\subsection{Graphs}

We consider in this article finite undirected graphs, without loops nor multiple edges.
Let $G = (V,E)$ be an $n$-vertex $m$-edge graph.
$V(G)$ (or simply $V$ if it is clear from the context) denotes the \define{set of vertices} of the graph $G$, while $E(G)$ (or simply $E$) denotes the \define{set of edges}.
For two vertices $x,y \in V$, we denote an edge between $x$ and $y$ by $xy$.
For a vertex $v \in V$, $N(v) = \{u \mid uv \in E\} $ denotes the \define{open neighborhood} of $v$, while $N[v] = N(v) \cup \{v\}$ denotes its \define{closed neighborhood}.
For a subset $S \subseteq V$, $N[S] = \bigcup_{v \in S} N[v]$ denotes the closed neighborhood of $S$.

The \define{incidence graph} $I(G)$ of a graph $G = (V,E)$ is a bipartite graph with $V \cup E$ as set of vertices, and for two vertices $v',e'$ of $I(G)$ corresponding respectively to a vertex $v$ and an edge $e$ of $G$, $v'$ is adjacent to $e'$ in $I(G)$ if $v$ is incident to $e$ in $G$, \ie\ $v$ is an endpoint of $e$.

\subsection{Tree-width}

A \define{tree-decomposition} (see \eg~\cite{Bo98}) of a graph $G$ is a tree $T$ in which each node $i \in T$ has an assigned set of vertices $X_i \subseteq V(G)$ (called \define{bag}), such that (1) every vertex $v \in V(G)$ appears in at least one bag $X_i$ of $T$, (2) every edge $uv \in V(G)$ has its both endpoints appearing in the same bag $X_j$ of $T$ (for some $j$), and (3) for every vertex $v \in V(G)$, the bags containing $v$ induce a connected subtree of $T$.
The width of a tree-decomposition is the size of the largest bag of $T$ minus one, \ie, $\max_{i \in T} |X_i| - 1$.
The \define{tree-width} of a graph $G$ is then the minimum width over all tree-decompositions of $G$.

\subsection{Parameterized complexity}

A problem is in \FPT\ (fixed-parameter tractable) parameterized by some parameter $k$ if it can be solved in $\ppO(f(k) \cdot p(n))$ time, for some computable function $f$ and a polynomial $p$, where $n$ is the size of the input; those problems are considered to be \emph{tractable} from the parameterized viewpoint.
Parameterized problems considered to be \emph{intractable} are at least \W{1}-hard, where \W{1} is one of the most important classes of the parameterized complexity theory and believed to be strictly including the class \FPT.
Finally, a problem is in \XP\ parameterized by some parameter $k$ if it can be solved in $\ppO(n^{f(k)})$ time, for some computable function $f$, where $n$ is the size of the input; the parameterized class \XP\ is known to strictly contain the parameterized class \W{1}, and hence also \FPT (see~\cite{FlGr06}).

See \eg~\cite{DoFe99,FlGr06} for some surveys on the parameterized complexity theory.

\subsection{Generalized domination}

Let $\si,\rh$ be two fixed subsets of $\N$ (where $\N$ denotes the set of nonnegative integers).
The \eSRDS\ problem (introduced in~\cite{Te94}, see also~\cite{Te94PhD,TePr97}) is defined as follows:

\problem{\eSRDS}{A graph $G = (V,E)$.}{Is there a subset $D \subseteq V$ such that for every $v \in D$, $|N(v) \cap D| \in \si$, and for every $v \notin D$, $|N(v) \cap D| \in \rh$? If so, $D$ is called a \SR-dominating set.}

A vertex $v$ is \define{dominated} by a vertex $u$ if $v \in N[u]$, and it is \define{dominated} by a set $S \subseteq V$ if $v \in N[S]$.
A subset of vertices $S \subseteq V$ is called a \define{dominating set} if every vertex of $G$ is dominated by $S$; a vertex $v$ in $S$ is said to be \define{selected}, while a vertex $v'$ not in $S$ is said to be \define{non-selected}.
A selected (\resp\ non-selected) vertex $v$ in a \SR-dominating set $D$ of a graph $G$ is \define{satisfied} with respect to \si\ (\resp\ \rh) if $|N_G(v) \cap D| \in \si$ (\resp\ $\in \rh$).

\section{Some general results on the parameterized complexity of \SRDS}\label{sec:general_complexity}

In the following, we give some general results on the parameterized complexity of \eSRDS, for instance when it is parameterized by the tree-width of the input graph, and when it is parameterized both by the tree-width of the input graph and the maximum size of a solution.

\subsection{When parameterized by the tree-width of the input graph}

Firstly, we consider the \SRDS\ problem when parameterized by the tree-width of the input graph, and prove that is in \XP\ for any polytime decidable recursive sets \si\ and \rh.
This result naturally motivates our study on the \emph{real} parameterized complexity class in which \eSRDS\ fall, depending on the sets \si\ and \rh.

\begin{theorem}
Let \si\ and \rh\ be two recursive sets of integers for which the membership of any integer $t$ can be computed in polynomial time.
Then \eSRDS\ (as well as minimization and maximization) is in \XP\ when parameterized by the tree-width of the input graph.
\end{theorem}
\begin{proof}
Let $G = (V,E)$ be the input graph of tree-width $\tw$.
Indeed, any \SR-dominating set will be of cardinality at most $|V|$.
Hence we can consider the problem \textsc{$[\si',\rh']$-Dominating Set} where $\si' = \si \cap \{0,\ldots,|V|\}$ and $\rh' = \rh \cap \{0,\ldots,|V|\}$ are both finite; note that $\si'$ and $\rh'$ can be computed in polynomial time, as we supposed that $\si$ and $\rh$ are computable in polynomial time.
Then inspired by algorithms solving \eSRDS\ in $\ppO(|\si| + |\rh|)^{f(tw)} \cdot n^{\ppO(1)}$ time when \si\ and \rh\ are both finite (see~\cite{RoBoRo09} and \cite[Theorem~3.2.1]{Ch11PhD}), one can easily construct an algorithm solving \eSRDS\ in \rfpt\ time when parameterized by the tree-width of the input graph.
\end{proof}

Using the well-known theorem from Courcelle \etal~\cite{CoMaRo01}, one can improve on the previous theorem when \si\ and \rh\ are both ultimately periodic sets, and prove that \eSRDS\ is in \FPT\ for such sets when parameterized by tree-width.

\begin{theoremnn}[\cite{CoMaRo01}]
Any graph problem expressible as an \CMSO\ logic formulae can be solved in \FPT\ time when parameterized by the tree-width of the input graph.
\end{theoremnn}

Hence to prove that \eSRDS\ is \FPT\ when parameterized by tree-width, it suffices to give a \CMSO\ logic formulae expressing this problem when given two ultimately periodic sets \si\ and \rh, as described in~\cite[Section~3.2.1]{Ch11PhD}.
However, this general theorem from Courcelle \etal~\cite{CoMaRo01} can only be used to quickly prove a problem to be \FPT, and does not provide a practical algorithm.

Using a careful analysis on the specific combinatorics of \eSRDS\ problem, van~Rooij \etal~\cite{RoBoRo09} proved that this problem can be solved in efficient \FPT\ time (\ie, single exponential in the parameter) when \si\ and \rh\ are finite or cofinite sets of integers; under the Strong Exponential Time Hypothesis~\cite{ImPa01}, and by a result from Lokshtanov \etal~\cite{LoMaSa11} on algorithms parameterized by tree-width, their algorithm is optimal.
In~\cite[Section~3.2.2]{Ch11PhD}, we extended the result of~\cite{RoBoRo09} to ultimately periodic sets.

\begin{theoremnn}[{\cite[Theorem~3.2.1]{Ch11PhD}}]
Let \si\ and \rh\ be two ultimately periodic sets of integers.
Then \eSRDS\ (as well as minimization and maximization) can be solved in efficient \FPT\ time when parameterized by the tree-width of the input graph.
\end{theoremnn}

\subsection{When parameterized both by the tree-width of the input graph and the maximum size of a solution}

We have seen that the parameterized complexity of \eSRDS\ highly depends on the sets \si\ and \rh, when it is parameterized solely by the tree-width of the input graph.
On the other hand, when parameterized both by the tree-width of the input graph and the maximum size of a solution, the problem becomes \FPT\ for any recursive sets \si\ and \rh.

\begin{theorem}
Let \si\ and \rh\ be two recursive sets of integers. Then \kSRDS\ is \FPT\ when parameterized both by the tree-width of the input graph and the maximum size $k$ of a solution.
\end{theorem}
\begin{proof}
In \kSRDS, we ask for a \SR-dominating set of cardinality at most $k$.
Hence every vertex of the input graph will have at most $k$ neighbors in the \SR-dominating set, so we can reduce this problem to \textsc{$[\si',\rh']$-Dominating Set} where $\si' = \si \cap \{0,\ldots,k\}$ and $\rh' = \rh \cap \{0,\ldots,k\}$ are both finite.
If $a(k)$ (\resp\ $b(k)$) denotes the maximum time needed to decide whether $t \in \si$ (\resp\ $t \in \rh$) for $t \leq k$, then $\si'$ and $\rh'$ can be computed in $\ppO(k \cdot a(k))$ time (not depending on the size of the input graph) and hence in \rfpt\ time.
Then again inspired by dynamic programming algorithms solving \eSRDS\ in $\ppO(|\si| + |\rh|)^{f(tw)} \cdot n^{\ppO(1)}$ time when \si\ and \rh\ are both finite (see~\cite{RoBoRo09} and \cite[Theorem~3.2.1]{Ch11PhD}), one can easily construct an algorithm solving \kSRDS\ in \rfpt\ time when parameterized both by the tree-width of the input graph and the maximum size $k$ of a solution.
\end{proof}

\section{Proof of Theorem~\ref{thm:main}: \W{1}-hardness of \eSRDSps}\label{sec:w1-hardness_first-step}

To prove Theorem~\ref{thm:main}, we first reduce from \kCDS\ to a variant of \eSRDS\ in which the input also contains a subset of \emph{preselected} vertices $W$ which are inconditionally included in the \SR-dominating set:

\paramproblem{\eSRDSps}{A graph $G = (V,E)$ of tree-width \tw, a subset $W \subseteq V$ of preselected vertices.}{\tw.}{Is there a set $D \subseteq V$, with $W \subseteq D$, such that for every $v \in V \setminus W$, if $v \in D$ then $|N(v) \cap D| \in \si$, and if $v \notin D$, then $|N(v) \cap D| \in \rh$?}

The problem \kCDS\ is known to be \W{1}-complete when parameterized by the tree-width of the input graph plus the maximum size of the capacitated dominating set, by a result from Dom \etal~\cite{DoLoSaVi08}:

\paramproblem{\kCDS}{A graph $G = (V,E)$ of tree-width \tw, a function $\capacity : V \rightarrow \N$, a positive integer $k$.}{k + \tw.}{Does $G$ admit a set $S$ of cardinality at most $k$ and a \define{domination function} \dom\ associating to each vertex $v \in S$ a set $\dom(v) \subseteq N(v) \setminus S$ of at most $\capacity(v)$ vertices, such that every vertex $w \in V \setminus S$ is in $\dom(v)$ for some $v \in S$? If so, $(S,\dom)$ is called a $k$-capacitated dominating set.}

In the following, we give an \rfpt-reduction from \kCDS\ to \eSRDSps, thus proving this latter problem to be \W{1}-hard when parameterized by the tree-width of the input graph.

Let \si\ be a set with arbitrarily large gaps between two consecutive elements (such that a gap of length at least $t$ is at distance $\poly(t)$ in \si), and let \rh\ be cofinite.
We define $q_0 = \min\{q \in \rh \mid \forall r \geq q, r \in \rh\}$, $\rh_{\min} = \min \{ q \mid q \in \rh \}$, and $\si_{\min} = \min \{ p \mid p \in \si \}$.
Furthermore, we suppose that $\rh_{\min} \geq 2$ (hence $q_0 \geq 2$) and $\si_{\min} \geq 1$.\footnote{We always suppose $0 \notin \rh$, as otherwise $S = \emptyset$ would be a trivial \SR-dominating set.}
The parameterized complexity of \eSRDSps\ (and more generally \eSRDS) for the extremal cases where $\rh_{\min} = 1$ or $\si_{\min} = 0$ is left open.

\subsection{Some functions on \si, and a technical condition}\label{subsec:functions_on_sigma}

In order to ease the description of our \rfpt-reduction, we suppose we are given the following computable functions on \si:
\begin{itemize}
	\item $\gapm(x,q)$: returns the lowest element $p$ of \si, greater than $q$, for which there are at least $x$ integers immediately \emph{before} $p$ which are not elements of \si;
	\item $\gapp(x,q)$: returns the lowest element $p$ of \si, greater than $q$, for which there are at least $x$ integers immediately \emph{after} $p$ which are not elements of \si;
	\item $\gapz(q)$: returns the lowest element $p$ of \si\ greater than $q$.
\end{itemize}
Intuitively, $\gapm$ can be read as ``gap before'' and $\gapp$ as ``gap after''; then $\gapm(x,q)$ (\resp\ $\gapp(x,q)$) is the minimum element $p > q$ of \si\ such that there is a gap of size at least $x$ right before (\resp\ right after) $p$.
Those functions will allow us to find some gaps of particular length in \si\ used in the construction of our gadgets for the \rfpt-reduction.

Some of our gadgets (for instance \gadgetname{capacity} and \gadgetname{limitation} gadgets, see below) will require gaps of length depending on the number of vertices in the input graph.
Thus we need a technical condition on $\gapp$, allowing us to get a polynomial-sized construction: we suppose that a gap of length at least $t$ can be found at distance $\poly(t)$ in \si, that is $\gapp(t,q)$ (for some $q \in \N$) is polynomial in $t$.
More formally, we suppose the following condition on $\gapp$:
\[\exists c \in \N, \forall t,q \in \N : \gapp(t,q) = \ppO(q + t^c)\]

\subsection{Overview of the construction}

Let $I(G)$ be the incidence graph of the original graph $G$.
We denote by $N_I(v)$ the neighborhood of $v$ in $I(G)$, corresponding to the set of edges incident to $v$ in $G$.
We call \define{original-vertex} a vertex in $I(G)$ which corresponds to a vertex in $G$, and \define{edge-vertex} a vertex in $I(G)$ which corresponds to an edge in $G$.

Given an instance $(G,\capacity,k)$ of \kCDS, where $G$ is of tree-width at most \tw\ and $\capacity$ is the \define{capacity function} on the vertices of $G$, we construct an instance $(H,W)$ of \eSRDSps\ with $H$ of bounded tree-width, such that $G$ admits a \kCDS\ if and only if $H$ admits a \SR-dominating set where vertices of $W$ are preselected.

Intuitively, a $k$-capacitated dominating set $S$ in $G$ should result in a \SR-dominating set $D$ with preselected vertices in $H$ as follows (see Figure~\ref{fig:encoding_of_capacitated_dominating_set}).
If a vertex of $G$ is in $S$, then the corresponding original-vertex of $H$ will be in $D$.
For each vertex $u \in V(G) \setminus S$ which is dominated by a vertex $v \in V(G)$, the edge-vertex $e$ in $I(G)$ representing the edge between $u$ and $v$ in $G$ will also be in $D$.
The \gadgetname{limitation} gadget~\gadget{L}{} will ensure that no more than $k$ original-vertices of $H$ are selected into $D$; those vertices will correspond to a $k$-capacitated dominating set $S$ in the original graph $G$.
The other vertices of $H$ (those which are not in $I(G)$) are vertices of the several gadgets, and may or may not be included in $D$.

\begin{figure}[ht]
	\begin{center}
		\subfigure[A $k$-capacitated dominating set in $G$.]{\includegraphics[height=3.5cm]{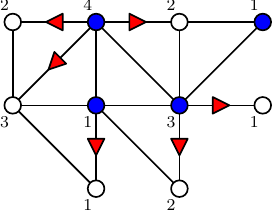}}
		\hspace{2cm}
		\subfigure[A \SR-dominating set in $I(G)$ (the gadgets are missing).]{\includegraphics[height=3.5cm]{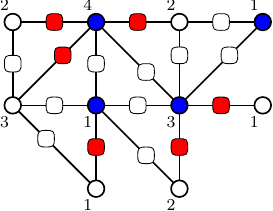}}
	\end{center}
	\caption{Encoding of a $k$-capacitated dominating set in $G$ as a \SR-dominating set in $I(G)$: blue vertices correspond to selected original-vertices, while red vertices correspond to selected edge-vertices. A number next to a vertex corresponds to its capacity in the \kCDS{} instance.}
	\label{fig:encoding_of_capacitated_dominating_set}
\end{figure}

The key idea of our \rfpt-reduction is that based on the \Gap\ functions, we can find sufficiently large gaps in \si, and use these gaps to control the number of selected neighbors of some particular vertices in the instance $(H,W)$, \eg, the number of selected neighbors of each selected original-vertex $v \in I(G)$ must not exceed its capacity in the corresponding $(G,\capacity,k)$ instance.

To construct $H$, we start with a copy of $I(G)$, and add some gadgets attached to original-vertices and edge-vertices of $I(G)$ as follows (see Figure~\ref{fig:overall_construction}):
\begin{itemize}
\item on each original-vertex $v \in I(G)$, we attach one \gadgetname{capacity} gadget~\gadget{C}{v} and one \gadgetname{satisfiability} gadget~\gadget{S}{v}; we also attach one \gadgetname{domination} gadget~\gadget{D}{v} which is also linked to all edge-vertices to which $v$ is adjacent;
\item on each edge-vertex $e \in I(G)$, we attach one \gadgetname{edge-selection} gadget~\gadget{E}{e};
\item we add a global \gadgetname{limitation} gadget~\gadget{L}{} with one \emph{central} vertex linked to every original-vertex of $I(G)$.
\end{itemize}

\begin{figure}[ht]
	\begin{center}
		\includegraphics[height=3cm]{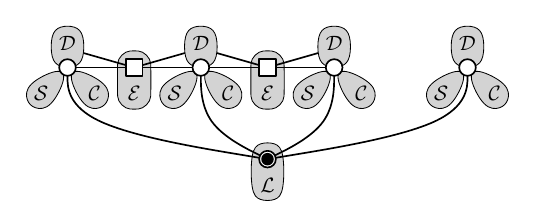}
	\end{center}
	\caption{Overview of the construction of the graph $H$. White vertices correspond to vertices of $I(G)$: white circles are original-vertices, and white squares are edge-vertices.}
	\label{fig:overall_construction}
\end{figure}

\subsection{Gadgets description and their correctness}

Before going into the details of our gadgets, let us introduce some notations.
In the following, a vertex of $H$ is said to be \emph{forced} if (by construction) it must be in any \SR-dominating set of $H$, \emph{choosable} if it can be satisfied no matter if it is in a \SR-dominating set of $H$ or not, and \emph{non-choosable} if it must not be in any \SR-dominating set of $H$.
Recall that a vertex is \emph{selected} if it is added to the considered \SR-dominating set, and \emph{preselected} vertices are selected vertices which are in the subset of preselected vertices $W$ given in the input of \eSRDSps.

Besides original- and edge-vertices, all our gadgets (almost) only contain these two kinds of vertices: forced or preselected vertices, and choosable vertices (see Lemma~\ref{lem:forced_and_choosable_vertices}); the only exception is for the \gadgetname{domination} gadget, which also contains a vertex which is ensured to be non-choosable.
Forced vertices are adjacent to a clique of size $\si_{\min}$, while choosable vertices are adjacent to $\gapz(q_0)$ forced vertices.
We now prove the following lemma, which provides necessary tools to prove the correctness of these gadgets.

\begin{lemma}\label{lem:forced_and_choosable_vertices}
Suppose that $H$ admits a \SR-dominating set $D$ with preselected vertices.
Then every forced vertex must be selected (\ie, be in $D$), and every choosable vertex is always satisfied no matter if it is selected or not.
\end{lemma}
\begin{proof}
By definition of the \eSRDSps\ problem, preselected vertices must be selected in any \SR-dominating set of $H$.

In every gadget, each forced vertex $v$ is included into a clique containing $\si_{\min} + 1$ vertices in total (including $v$), forcing it to be selected.
Indeed, as $D$ is a \SR-dominating set of $H$ with preselected vertices, and $0 \notin \rh$, at least one vertex $u$ in this clique must be selected (it may be $v$).
Then $u$ needs at least $\si_{\min}$ selected neighbors, and hence all vertices of the clique, including $v$, must be selected.

Finally, each choosable vertex has $\gapz(q_0) \in \si \cap \rh$ selected neighbors in total (forced or preselected), which allows it to be satisfied no matter if it is selected or not in the \SR-dominating set $D$ of $H$ with preselected vertices.

This completes the proof of the lemma.
\end{proof}

We now describe the gadgets used for the \rfpt-reduction, and prove each of them to be correct.
In the corresponding gadgets' figures, black triangles represent preselected vertices, black cliques and plain disks represent forced vertices, circles represent choosables vertices, and crosses represent non-choosables vertices.

\subsubsection{Domination gadget}

For each original-vertex $v \in I(G)$, we add to $H$ one \gadgetname{domination} gadget~\gadget{D}{v}.
This gadget ensures that an original-vertex of $I(G)$ is either selected, or it has at least one selected neighbor in $I(G)$, \ie, the selected original- and edge-vertices form a dominating set of $I(G)$.
The gadget is constructed as follows (see Figure~\ref{fig:gadget_D}): given an original-vertex $v \in I(G)$, we add a non-choosable vertex $v'$ linked to $v$ and to each edge-vertex $e \in N_I(v)$, $q_0 - 2$ independent preselected vertices linked to $v'$, and an extra forced vertex $v''$ linked to $v'$ with $\gapp(1,\si_{\min}) - \si_{\min}$ independent preselected neighbors and to a clique with $\si_{\min}$ vertices.

\begin{figure}[ht]
	\begin{center}
		\includegraphics[height=3.1cm]{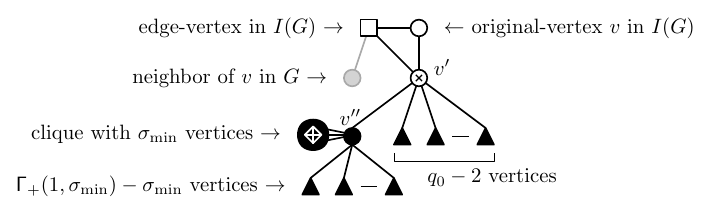}
	\end{center}
	\caption{\gadgetname{Domination} gadget~\gadget{D}{v}}
	\label{fig:gadget_D}
\end{figure}

\begin{lemma}\label{lem:gadget_D}
If $H$ admits a \SR-dominating set with preselected vertices, then either the original-vertex $v$ is selected, or at least one of its edge-neighbors is selected.
\end{lemma}
\begin{proof}
Suppose $H$ admits a \SR-dominating set with preselected vertices.

The forced vertex $v''$ has $\gapp(1,\si_{\min})$ forced and preselected neighbors, and its only other neighbor is $v'$.
Hence $v'$ cannot be selected, as otherwise $v''$ would have $\gapp(1,\si_{\min}) + 1 \notin \si$ selected neighbors and hence will not be satisfied, which would contradict the fact that $H$ admits a \SR-dominating set with preselected vertices.

Now notice that $v'$ has $q_0 - 1$ preselected and forced neighbors: $q_0 - 2$ preselected neighbors, and one more forced neighbor $v''$.
Recall that $q_0 - 1 \notin \rh$.
Hence, it must have at least one more selected neighbor in order for $v'$ to be satisfied, and this neighbor cannot be $v''$ as proved before; its only other selectable neighbors are $v$, and the edge-neighbors of $v$ to which $v'$ is adjacent.
So either $v$ must be selected, or at least one of these edge-neighbors must be selected.
\end{proof}

\subsubsection{Edge-selection gadget}

For each edge-vertex $e \in I(G)$, we add to $H$ one \gadgetname{edge-selection} gadget~\gadget{E}{e}.
The selected edge-vertices in $I(G)$ will correspond to the domination function of the \kCDS\ problem we reduce from.
This gadget ensures that if the edge-vertex $e$ is selected, then so is at least one of its original-neighbors in $I(G)$.
The gadget is constructed as follows (see Figure~\ref{fig:gadget_E}): given an edge-vertex $e \in I(G)$, we add $\gapm(1,q_0+1) - 1$ independent preselected vertices linked to $e$.

\begin{figure}[ht]
	\begin{center}
		\includegraphics[height=2.4cm]{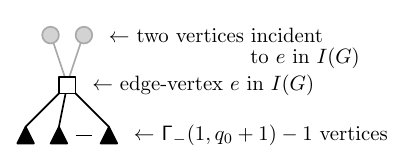}
	\end{center}
	\caption{\gadgetname{Edge-selection} gadget~\gadget{E}{e}}
	\label{fig:gadget_E}
\end{figure}

\begin{lemma}\label{lem:gadget_E}
If $H$ admits a \SR-dominating set with preselected vertices, then if the edge-vertex $e$ is selected, then so is at least one of its original-neighbors. 
\end{lemma}
\begin{proof}
Suppose $H$ admits a \SR-dominating set with preselected vertices in which the edge-vertex $e$ is selected.

Then, $e$ has exactly $\gapm(1,q_0+1) - 1$ preselected neighbors in $H$, all being from its associated \gadgetname{edge-selection} gadget~\gadget{E}{e}.
If $e$ is selected, then it needs at least one more selected neighbor in $H$, corresponding to an original-vertex in $I(G)$, as $\gapm(1,q_0+1) - 1 \notin \si$ and $\gapm(1,q_0+1) \in \si$.
Notice that if $e$ is not selected, then it is satisfiable as it then has $\gapm(1,q_0+1) - 1 \in \rh$ selected neighbors in this gadget.
\end{proof}

\subsubsection{Satisfiability gadget}

For each original-vertex in $v \in I(G)$, we add to $H$ one \gadgetname{satisfiability} gadget~\gadget{S}{v}.
This gadget allows any non-selected original-vertex in $I(G)$ to have a valid number of selected neighbors in $H$ with respect to \rh, \ie, it can have at least $q_0 \in \rh$ selected neighbors.
The gadget is constructed as follows (see Figure~\ref{fig:gadget_S}): given an original-vertex $v \in I(G)$, we add $q_0$ independent choosable vertices with $\gapz(q_0)$ independent preselected neighbors each.

\begin{figure}[ht]
	\begin{center}
		\includegraphics[height=3.0cm]{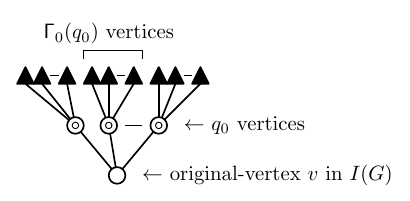}
	\end{center}
	\caption{\gadgetname{Satisfiability} gadget~\gadget{S}{v}}
	\label{fig:gadget_S}
\end{figure}

%
%

\subsubsection{Capacity gadget}

For each original-vertex in $v \in I(G)$, we add to $H$ one \gadgetname{capacity} gadget~\gadget{C}{v}.
This gadget ensures that a selected original-vertex $v$ of $I(G)$ has at most $\capacity(v)$ selected neighbors in $I(G)$, and is satisfied no matter how many of its edge-neighbors are selected.
Moreover, this gadget allows a selected original-vertex to have a valid number of selected neighbors in $H$ with respect to \si.
The gadget is constructed as follows (see Figure~\ref{fig:gadget_C}): given an original-vertex $v \in I(G)$, we add $\gapp\big(\ppdeg_G(v)+q_0, \capacity(v)\big) - \capacity(v) - 1$ independent preselected vertices linked to $v$, and $\capacity(v)$ independent choosable vertices linked to $v$ with $\gapz(q_0+1) - 1$ independent preselected neighbors each.

\begin{figure}[ht]
	\begin{center}
		\includegraphics[height=4.2cm]{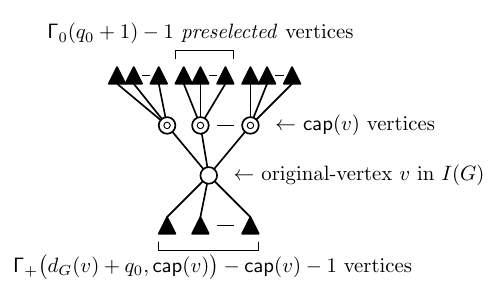}
	\end{center}
	\caption{\gadgetname{Capacity} gadget~\gadget{C}{v}}
	\label{fig:gadget_C}
\end{figure}

\begin{lemma}\label{lem:gadget_C}
If $H$ admits a \SR-dominating set with preselected vertices in which the original-vertex $v$ is selected, then $v$ has at most $\capacity(v)$ selected edge-neighbors in $I(G)$.
\end{lemma}
\begin{proof}
Suppose $H$ admits a \SR-dominating set with preselected vertices in which the original-vertex $v$ is selected.

Then, it has $\gapp\big(\ppdeg_G(v)+q_0, \capacity(v)\big) - \capacity(v) - 1$ forced neighbors in the \gadgetname{capacity} gadget~\gadget{C}{v} associated to it.
The other neighbors it has in $H$ are edge-vertices in $I(G)$, the central forced vertex $c$ of the \gadgetname{limitation} gadget~\gadget{L}{}, and an independent set of $\capacity(v) + q_0 > \capacity(v)$ choosable vertices from \gadgetname{capacity}~\gadget{C}{v} and \gadgetname{satisfiability}~\gadget{S}{v} gadgets which can be selected if $v$ has less than $\capacity(v)$ selected edge-neighbors in $I(G)$.

Hence it can have at most $\capacity(v)$ more selected neighbors in $H$, as otherwise by definition of $\gapp$ it would need $\ppdeg_G(v) + q_0 + 1 + \capacity(v) + 1$ selected vertices in total to be satisfied, which is $1$ more than its overall degree in $H$.
\end{proof}

\subsubsection{Limitation gadget}

We add to $H$ one global \gadgetname{limitation} gadget~\gadget{L}{}.
This gadget limits the number of selected original-vertices in $I(G)$ to at most $k$ vertices, where $k$ is the parameter of the original \kCDS\ problem we reduce from.
The gadget is constructed as follows (see Figure~\ref{fig:gadget_L}): for the whole graph $H$, we add one central forced vertex $c$ linked to every original-vertex of $I(G)$ and to a clique with $\si_{\min}$ vertices, $\gapp\big(|V(G)|+k,k + \si_{\min}\big) - k - \si_{\min}$ independent preselected vertices linked to $c$, and $k$ independent choosable vertices linked to $c$ with $\gapz(q_0) - 1$ independent preselected neighbors each.

\begin{figure}[ht]
	\begin{center}
		\includegraphics[height=4.2cm]{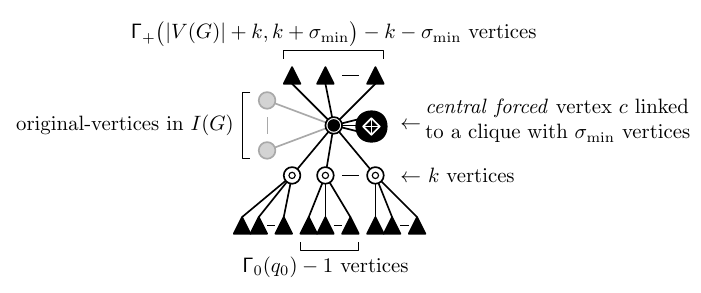}
	\end{center}
	\caption{\gadgetname{Limitation} gadget~(\gadget{L}{})}
	\label{fig:gadget_L}
\end{figure}

\begin{lemma}\label{lem:gadget_L}
If $H$ admits a \SR-dominating set with preselected vertices, then at most $k$ original-vertices are selected.
\end{lemma}
\begin{proof}
Suppose $H$ admits a \SR-dominating set $D$ with preselected vertices.

Notice first that by Lemma~\ref{lem:forced_and_choosable_vertices}, the central vertex $c$ is forced and hence must be selected in the \SR-dominating set $D$ with preselected vertices.
The vertex $c$ has $\gapp\big(|V(G)| + k, k + \si_{\min}\big) - k$ selected neighbors in the \gadgetname{limitation} gadget~\gadget{L}{}, including the $\si_{\min}$ forced vertices from the clique, and an independent set of $k$ choosable neighbors, which may be selected in $D$.

If at least $k+1$ original-vertices are selected, then the number of selected neighbors of $c$ will be at least $\gapp\big(|V(G)| + k, k + \si_{\min}\big) + 1$, and can be at most $\gapp\big(|V(G)| + k, k + \si_{\min}\big) + |V(G)| + k$ if all the original-vertices and all the choosable neighbors of $c$ are selected.
However, by definition of $\gapp$ and due to the gap of length $|V(G)| + k$, none of these numbers is in \si, and hence none of them allows $c$ to be satisfied, which is a contradiction to the fact that $D$ is a \SR-dominating set of $H$ with preselected vertices.
Notice that if at most $k$ original-vertices are selected, then some of its choosable neighbors may be selected in $D$ in order for $c$ to be satisfied with $\gapp\big(|V(G)| + k, k + \si_{\min}\big) \in \si$ selected neighbors in total.

Hence at most $k$ original-vertices are selected in $D$.
\end{proof}

\subsubsection{Correctness of the overall construction}

We now deserve the remaining of this section to the proof of correctness of the overall construction, that is $(G,\capacity,k)$ is a yes-instance of \kCDS\ if and only if $(H,W)$ is a yes-instance of \eSRDSps, and that the reduction is \rfpt.

Let us first prove that we can construct a yes-instance of \eSRDSps\ from a yes-instance of \kCDS, using our construction described above.

\begin{lemma}\label{lem:construction_of_a_solution}
Let $(S,\dom)$ be a $k$-capacitated dominating set of the original graph $G$.
Given $I(G)$ the incidence graph of $G$, we construct $D_0 \subseteq V(H)$ as follows:
\begin{itemize}
\item for each selected vertex $v_G \in V(G) \cap S$, we add the corresponding original-vertex $v_H \in V(H)$ to $D_0$;
\item for each dominated vertex $u_G \in V(G) \setminus S$ such that $u \in \dom(v)$ for some selected vertex $v_G \in V(G) \cap S$, we add the corresponding edge-vertex $e_H$ to $D_0$, where $e_H$ is the edge-vertex adjacent to the original vertices $u_H,v_H \in V(H)$.
\end{itemize}
Then $D_0$ can be extended to a \SR-dominating set $D$ of $H$ with preselected vertices.
\end{lemma}
\begin{proof}
Suppose that the original graph $G$ admits a $k$-capacitated dominating set $(S,\dom)$, and let $D_0$ be constructed as described in the lemma statement.
We prove that $D_0$ can be extended to a \SR-dominating set $D$ with preselected vertices to the whole constructed graph $H$.

Notice that each gadget contains two types of vertices: forced or preselected vertices, and choosable vertices (excepted for \gadgetname{domination} gadget, which also contains a non-choosable vertex).
By definition of the \eSRDSps\ problem, preselected vertices must be selected in $D$.
By Lemma~\ref{lem:forced_and_choosable_vertices}, every forced vertex must be selected in $D$, and every choosable vertex can be satisfied no matter if it is selected or not in $D$.
Finally, the non-choosable vertex $v'$ in each \gadgetname{domination} gadget has exactly $q_0 \in \rh$ selected neighbors: $q_0 - 2$ preselected neighbors, $1$ forced neighbor ($v''$), and exactly one of the original-vertex or edge-vertex selected in its neighborhood (depending on whether the original-vertex is selected or not) by construction of $D_0$.

Now let us show that original- and edge-vertices can also be satisfied.
Every original-vertex has $\gapp\big(\ppdeg_G(v)+q_0, \capacity(v)\big) - \capacity(v)$ forced neighbors (from \gadgetname{capacity}~\gadget{C}{v} and \gadgetname{limitation}~\gadget{L}{} gadgets) and $\capacity(v) + q_0$ choosable vertices (from \gadgetname{satisfiability}~\gadget{S}{v} gadget) in $H$, which altogether allows the original-vertex to have $\gapp\big(\ppdeg_G(v)+q_0, \max\{\capacity(v),q_0\}\big) \in \si \cap \rh$ selected neighbors by selecting some choosable vertices if necessary; hence each original-vertex is satisfiable no matter if it is selected or not.
By construction of the \gadgetname{edge-selection} gadget, each edge-vertex is also satisfiable (see Lemma~\ref{lem:gadget_E}).

Altogether, this proves that all the vertices of the constructed graph $H$ can be satisfied by extending $D_0$ to a \SR-dominating set $D$ of $H$ with preselected vertices.
\end{proof}

As a corollary of Lemma~\ref{lem:construction_of_a_solution}, we obtain the ``only if'' part of our statement:

\begin{corollary}\label{cor:equivalence_of_solution_only_if}
If the input graph $G$ admits a $k$-capacitated dominating set, then the constructed graph $H$ admits a \SR-dominating set with preselected vertices.
\end{corollary}

It remains to show the ``if'' part of our statement:

\begin{lemma}\label{lem:equivalence_of_solution_if}
If the constructed graph $H$ admits a \SR-dominating set with preselected vertices, then the input graph $G$ admits a $k$-capacitated dominating set.
\end{lemma}
\begin{proof}
Let $D$ be a \SR-dominating set of $H$ with preselected vertices.
Let $H_v$ be the original-vertices of $H$ corresponding to the vertices of the original graph $G$, and let $D_v \subseteq D$ (\resp\ $D_e \subseteq D$) be the selected original-vertices (\resp\ edge-vertices) in $H$ which arise from vertices of $G$ (\resp\ from edges of $G$).
Note that $D_v \subseteq H_v$.
By Lemma~\ref{lem:gadget_D} on \gadgetname{domination} gadget, each original-vertex in $H_v$ is either selected (\ie, is in $D_v$) or has a selected neighbor in $D_e$, and hence $D_v \cup D_e$ forms a dominating set in $I(G)$.
By Lemma~\ref{lem:gadget_E} on \gadgetname{edge-selection} gadget, each edge-vertex in $D_e$ has at least one selected original-neighbor in $D_v$.
By Lemma~\ref{lem:gadget_C} on \gadgetname{capacity} gadget, each selected original-vertex $v$ in $H_v$ (\ie, $v \in D_v$) has at most $\capacity(v)$ selected edge-neighbors in $D_e$.
For a selected original-vertex $v$, we set $\dom(v) = \{u \mid u \in N_H(v) \cap D_e\}$.
Finally, by Lemma~\ref{lem:gadget_L} on \gadgetname{limitation} gadget, $|D_v| \leq k$.
Then $(D_v,\dom)$ is a $k$-capacitated dominating set of the original $G$, where $\dom$ is the domination function of vertices in $G$.
\end{proof}

\begin{lemma}\label{lem:fpt_reduction}
The reduction from \kCDS\ to \eSRDS\ is an \rfpt-reduction.
More precisely, the graph $H$ can be constructed in $\poly\big(|V(G)|\big)$ time, and $\tw(H) = \max\big\{2 \tw(G) + 1, \si_{\min} + 1\big\}$.
\end{lemma}
\begin{proof}
Let \si\ and \rh\ be fixed. Let $G$ be the original input graph of \kCDS\ and $H$ be the constructed graph for \eSRDSps.

First, let us prove that the size of $H$ is polynomial in the size of $G$.
The \gadgetname{limitation} gadget is created once for the whole graph, and a \gadgetname{capacity} gadget is created once for each original-vertex.
The cardinalities of those two gadgets depend on the $\Gap$ functions, and a polynomial in the size of $G$ (which justifies our technical constraint on \si, see Section~\ref{subsec:functions_on_sigma}).
The other gadgets are created once for each original-vertex or edge-vertex of $I(G)$, and each has a cardinality depending on the $\Gap$ functions and on \si\ and \rh.
Hence the overall number of vertices in $H$ depends only on \si\ and \rh, and polynomially on the number of vertices and edges of $G$.
The reduction is then polynomial in the size of the original graph $G$.

Let us now prove that the tree-width of $H$ is polynomial in the tree-width of $G$.
Let $T(G)$ be an optimal tree-decomposition of $G$.
We first construct a tree-decomposition $T\big(I(G)\big)$ of $I(G)$ of width at most $\max\{\tw(G), 2\}$ as follows.
Start with $T\big(I(G)\big) = T(G)$.
Then, for every edge-vertex $e_{uv}$ in $I(G)$, we attach a pendant bag $\{u,v,e_{uv}\}$ to one of the bags of $T\big(I(G)\big)$ containing both $u$ and $v$.
Indeed, the obtained $T\big(I(G)\big)$ is a tree-decomposition of $I(G)$ of width at most $\max\{\tw(G), 2\}$.

It remains to explain how one can construct a tree-decomposition $T(H)$ of $H$ of width bounded in terms of the tree-width of $I(G)$.
By construction, the \gadgetname{edge-selection}, \gadgetname{satisfiability} and \gadgetname{capacity} gadgets are trees (hence of tree-width $1$), while \gadgetname{domination} and \gadgetname{limitation} gadgets are trees with an additional clique of size $\si_{\min}$ attached to one of its nodes (hence of tree-width $\si_{\min} + 1$).
Therefore we start with $T(H) = T\big(I(G)\big)$, and consider an optimal rooted tree-decomposition of each gadget, which includes the original-vertex and/or edge-vertex to which the gadget is attached in $H$.
Then, for each gadget, we add its tree-decomposition to $T(H)$, and link its root to one of the bags containing its associated original- or edge-vertex; after this step, $T(H)$ is of width at most $\max\{\tw(G), \si_{\min} + 1\}$.
For every \gadgetname{domination} gadget, we also add the vertex $v'$ of this gadget to each bag in $T(H)$ containing its associated original-vertex; this may increase the width of $T(H)$ by at most a factor of $2$ (it is equivalent to replacing each original-vertex in the tree-decomposition by two vertices).
Finally, we add the central vertex $c$ of the \gadgetname{limitation} gadget to every bag containing an original-vertex; this may increase the width of $T(H)$ by at most $1$.

We end with a tree-decomposition $T(H)$ of $H$ of width $\max\big\{2 \tw(G) + 1, \si_{\min} + 1\big\}$, which is polynomial in the tree-width of $G$ (as $\si$ is fixed).
\end{proof}

This completes the proof of the first step of our reduction, that is \eSRDSps\ is \W{1}-hard when parameterized by tree-width.

\section{Proof of Theorem~\ref{thm:main}: \W{1}-hardness of \eSRDS}\label{sec:w1-hardness_second-step}

We know describe the second step of our reduction, from \eSRDSps\ to \eSRDS, and complete the proof of Theorem~\ref{thm:main}.

\subsection{Reduction.}

In \eSRDSps\ problem, some vertices are preselected in a potential \SR-dominating set, are all of degree $1$, and are supposed to be satisfied no matter how the constraints from $\si$ are.
For this second step of our reduction, it suffices to show how these vertices can be forced to be selected and satisfied for our problem \eSRDS, while respecting the constraints from $\si$.
Notice that if $\si_{\min} + 1 \in \si$, then this reduction is easy: we link each of these vertices to a clique with $\si_{\min}$ vertices, using the same arguments as for gadget \gadgetname{limitation}~\gadget{L}{} (see the proof of Lemma~\ref{lem:gadget_L}).

\bigskip

Let $(G,W)$ be an instance of \eSRDSps.
We construct a graph $H$ for \eSRDS, such that vertices in $H$ corresponding to the subset $W$ of $V(G)$ are all forced to be selected, and are satisfied with respect to constraints from $\si$.
Notice that the sets $\si$ and $\rh$ are not modified during this reduction.
Moreover, the goal of our reduction is to prove that \eSRDS\ is \W{1}-hard when parameterized by tree-width.
Thus we also have to ensure that the tree-width of the constructed graph $H$ is bounded in the tree-width of the initial graph $G$.

Let $\alpha$ be the smallest integer which is both in $\si$ and $\rh$, that is $\alpha = \min \{q \mid q \in \si \cap \rh\}$, and let $\beta$ be the length of the gap between the first two integers in $\si$, that is $\beta = \min \{p - 1 \mid \si_{\min} + p \in \si\}$.
As $\si$ is infinite and $\rh$ is cofinite, these two integers are well-defined.

We now describe the construction of the graph $H$.
Initially, $H$ is composed of $\alpha$ copies $G_1,\ldots,G_{\alpha}$ of the initial graph $G$.
For each copy $v_i$ in $G_i$ of a preselected vertex $v \in W$ of $G$, with $1 \leq i \leq \alpha$, we add to $H$ a clique $K_i$ with $\si_{\min}$ vertices, and linked to $v_i$.
Finally, for each preselected vertex $v$ of $G$, we add an independent set of $\beta$ vertices in $H$, and link these vertices to all the copies $v_i$ of $v$ (see Figure~\ref{app_fig:second_step_construction}).

\begin{figure}[h]
	\begin{center}
		\includegraphics[height=3.9cm]{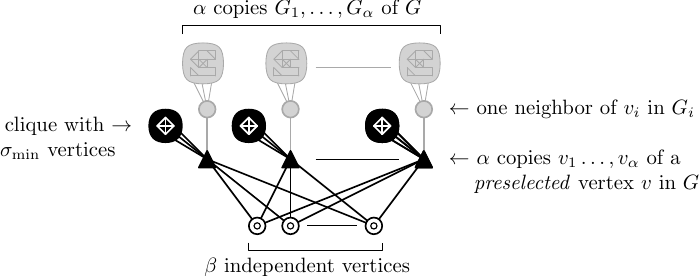}
	\end{center}
	\caption{Overall construction used to reduce from \eSRDSps\ to \eSRDS.}
	\label{app_fig:second_step_construction}
\end{figure}

\newpage

\subsection{Correctness.}

We now prove that the construction is correct, that is:
\begin{itemize}
	\item if $H$ admits a \SR-dominating set $D_H$, then each copy $v_i \in V(H)$ of a preselected vertex $v$ of $G$ must be selected in $D_H$;
	\item $(G,W)$ is a yes-instance for \eSRDSps\ if and only if $H$ is a yes-instance for \eSRDS;
	\item this is an \rfpt-reduction.
\end{itemize}

Notice that the vertices of each clique with $\si_{\min}$ vertices are satisfied, for the same arguments as for gadget \gadgetname{limitation}~\gadget{L}{} (see the proof of Lemma~\ref{lem:gadget_L}).

\begin{lemma}\label{lem:satisfied_vertices}
If $H$ admits a \SR-dominating set $D_H$, then each copy $v_i \in V(H)$ of a preselected vertex $v$ of $G$ must be selected in $D_H$.
\end{lemma}
\begin{proof}
Let $v_i \in V(H)$ be a copy of a preselected vertex $v \in W$ of degree $1$ in $G$, with $1 \leq i \leq \alpha$.
By construction, this vertex is linked to a clique with $\si_{\min}$ vertices.
For the same arguments as for gadget \gadgetname{limitation}~\gadget{L}{} (see the proof of Lemma~\ref{lem:gadget_L}), the vertex $v_i$ must be selected in $D_H$.
\end{proof}

\begin{lemma}
$(G,W)$ is a yes-instance for \eSRDSps\ if and only if $H$ is a yes-instance for \eSRDS, that is the initial graph $G$ admits a \SR-dominating set $D_G$ containing the subset $W \subseteq V(G)$ of preselected vertices if and only if the constructed graph $H$ admits a \SR-dominating set $D_H$.
\end{lemma}
\begin{proof}
Suppose first that $G$ admits a \SR-dominating set $D_G$, containing the subset $W \subseteq V(G)$ of preselected vertices all of degree $1$.

Let us explain how to construct a \SR-dominating set $D_H$ of $H$.
Initially, the set $D_H$ is empty.
For each preselected vertex $v$ in $G$, we add to $D_H$ all its copies $v_1,\ldots,v_{\alpha}$ in $H$.
For each vertex $u \in V(G)$ which is in the \SR-dominating set $D_G$, we add to $D_H$ all its copies $u_1,\ldots,u_{\alpha}$ in $H$, that is $u_1,\ldots,u_{\alpha} \in D_H$, and the independent set of $\beta$ vertices.
For each vertex $u' \in V(G)$ which is not in the \SR-dominating set $D_G$, none of its copies $u'_1,\ldots,u'_{\alpha}$ in $H$ are added to $D_H$, that is $u'_1,\ldots,u'_{\alpha} \notin D_H$, nor are the independent set of $\beta$ vertices.
As $D_G$ is a \SR-dominating set of $G$, and the selected neighborhood in $H$ of each copy $u_i$ is the same, all theses vertices are satisfied.

We now show that $D_H$ is a \SR-dominating set of $H$.
Recall that in $G$, the vertex $v$ from which $v_i$ is a copy is of degree $1$, and hence it as only one neighbor in $G$ which we noted $u$.
Hence the only neighbors of $v_i$ in $H$ are the copy $u_i$ of $u$, the vertices of the clique with $\si_{\min}$ vertices, and the independent set of $\beta$ vertices.
Consider the two possible cases for its neighbor $u_i$.
If $u_i \notin D_H$, then $v_i$ is satisfied if its only neighbors are the vertices in the clique with $\si_{\min}$ vertices, as in this case it has $\si_{\min} \in \si$ selected neighbors in $H$.
If $u_i \in D_H$, then $v_i$ is satisfied if all its neighbors in $H$ are also in $D_H$, that is the clique with $\si_{\min}$ vertices, the independent set of $\beta$ vertices and the vertex $u_i$, as in this case it has $\si_{\min} + \beta + 1 \in \si$ selected neighbors in $H$.
Notice that by construction of $H$, $u_i$ is always satisfied, as $D_G$ is a \SR-dominating set of $G$.
Finally, notice that the $\beta$ independent vertices have exactly $\alpha \in \si \cap \rh$ selected neighbors, which are the vertices $v_1,\ldots,v_{\alpha}$, and hence they are always satisfied no matter if they are selected or not.
Thus the set $D_H$ we have constructed is a \SR-dominating set of $H$.

\medskip

Suppose now that $H$ admits a \SR-dominating set $D_H$.
We show how to construct a \SR-dominating set $D_G$ of $G$, containing the subset $W \subseteq V(G)$ of preselected and satisfied vertices.
By Lemma~\ref{lem:satisfied_vertices}, all the copies $v_i$ in $H$ of a preselected vertex $v \in W$ in $G$ are selected, and as $D_H$ is a \SR-dominating set of $H$, they are all satisfied.
Let $G_i$ be one of the $\alpha$ copies of $G$ in the graph $H$.
Then $D_G = D_H \cap V(D_{G_i})$, that is we consider only the vertices of $D_H$ which are in the copy $G_i$ of $G$.
As $D_H$ is in particular a \SR-dominating set of $G_i$, the set $D_G$ is a \SR-dominating set of $G$ containing the subset $W \subseteq V(G)$ of preselected vertices.
\end{proof}

\begin{lemma}
The reduction is an \rfpt-reduction: the graph $H$ is of size polynomial in the size of $G$, and $\tw(H) \leq (\si_{\min} + \beta) \cdot \tw(G)$.
\end{lemma}
\begin{proof}
Recall that $\si$ and $\rh$ are two fixed sets of integers, which are part of the definition of \eSRDS\ (they are not parameters).
Let $(G,W)$ be the instance of \eSRDSps, and let $H$ be the constructed instance for \eSRDS.

Let $n = |V(G)|$.
The graph $H$ is contains $\alpha$ copies of $G$, to which are added at most $n$ cliques with $\si_{\min}$ vertices each (one clique for each preselected vertex in each copy of $G$).
Finally, at most $n$ independent set of $\beta$ vertices each are added to $H$, one for each set of copies of a preselected vertex.

Thus we have $|V(H)| \leq |V(G)| \cdot \big((\alpha \cdot \si_{\min}) + \beta\big)$.

\bigskip

We now show that the tree-width of $H$ is bounded in the tree-width of $G$.
For this purpose, suppose that $(T_G,\chi_G)$ is a tree-decomposition of $G$, of with $\tw(G)$, and note $(T_H,\chi_H)$ the tree-decomposition of $H$ we will construct.

First, let $T_H = T_G$.
For each bag $X_i \in \chi_G$, we construct the corresponding bag $X'_i \in \chi_H$ by adding to $X'_i$ all the copies $u_i \in V(H)$ of each vertex $u \in V(G)$ which are in the bag $X_i$.
Moreover, if $u \in V(G)$ is a preselected vertex, then we also add to $X'_i$ all the vertices of each clique with $\si_{\min}$ vertices (linked to each copy $u_i \in V(H)$ of the vertex $u \in V(G)$), and all the vertices of the independent set of $\beta$ vertices linked to all these copies in $H$.
It is easy to verify that we obtain a tree-decomposition of $H$, of width $\tw(H) \leq (\si_{\min} + \beta) \cdot \tw(G)$.
\end{proof}

\section{Conclusion}\label{sec:conclusion}

While the restriction to graphs of bounded tree-width is usually used to reach \FPT\ parameterized complexity, we have proven that for (infinitely) many cases of $\si$ and $\rh$ the \eSRDS\ problem becomes \W{1}-hard when parameterized by tree-width, for instance when $\si$ contains arbitrary large gaps, and $\rh$ is cofinite.
Associated with the proof that this problem is \FPT\ parameterized by tree-width when \si\ and \rh\ are ultimately periodic sets , we are getting closer to a complete dichotomy of the parameterized complexity of \eSRDS\ parameterized by tree-width.

This naturally requests further investigations of the parameterized complexity of this problem for other cases of sets \si\ and \rh.
In particular, what is the parameterized complexity of \eSRDS\ when the sets are recursive with bounded gaps between two consecutive elements but not ultimately periodic?
Can we circumvent the technical constraint on \si\ which requires a gap of length at least $t$ to be at distance $\poly(t)$ in \si?
Does the \W{1}-hardness of \eSRDS\ depend on the properties of $\si$, or there exists cases of $\rh$ for which it is also \W{1}-hard?

\subsubsection{Acknowledgement.} The author wishes to thank M. Liedloff, I. Todinca and several referees for constructive remarks on this work.

\small
\bibliographystyle{model1b-num-names}
\bibliography{../../../../../../Bibliography/Biblio}

\end{document}